\documentclass[a4paper,11pt,english,spanish,leqno]{amsart}
\usepackage[utf8x]{inputenc}
\usepackage[bitstream-charter]{mathdesign}
\usepackage{amsmath, xcolor,amsthm,epsfig,epstopdf,url,array}
\usepackage[colorlinks=true]{hyperref}
\usepackage[all]{xy}
\usepackage{tikz}
\tikzset{node distance=2cm, auto}
\usepackage{subcaption}
\usepackage[backrefs]{amsrefs}
\usepackage{pgf,tikz}
\usetikzlibrary{arrows}
\usepackage{float}

\theoremstyle{plain}
\newtheorem{thm}{Theorem}[section]

\newtheorem{lem}[thm]{Lemma}
\newtheorem{prop}[thm]{Proposition}

\newtheorem{cor}[thm]{Corollary}

\theoremstyle{definition}
\newtheorem{defn}[thm]{Definition}

\theoremstyle{remark}
\newtheorem*{dem}{Proof}

\theoremstyle{plain}

\def\S#1{\mathbb{S}^{#1}}
\def\Esp#1#2{\text{\rm E}_{#1}\left[#2\right]}

\def\E#1{\mathcal{E}_{#1}}

\def\R{\mathbb{R}}

\def\i{\mathbf{i}}

\begin{document}

\title[The Diamond ensemble]{The Diamond ensemble: a constructive set of points with small logarithmic energy}

\author{Carlos Beltr\'an and Uju\'e Etayo}

\date{\today{}}

\thanks{The authors would like to thank Edward Saff and Peter Grabner for many enriching discussions on this topic.
This research has been partially supported by Ministerio de Economía y Competitividad, Gobierno de España, through grants MTM2017-83816-P,  MTM2017-90682-REDT, and by the Banco de Santander and Universidad de Cantabria grant 21.SI01.64658.}

\subjclass[2010]{31C12, 31C20, 41A60, 52C15, 52C35, 52A40.}

\keywords{Minimal logarithmic energy, constructive spherical points.}

\address{Departamento de Matem\'aticas, Estad\'{\i}stica y Computaci\'on, 
Universidad de Cantabria,  Fac. Ciencias, Avd. Los Castros s/n, 39005 Santander, Spain}
\email{etayomu@unican.es}

\address{Departamento de Matem\'aticas, Estad\'{\i}stica y Computaci\'on, 
Universidad de Cantabria,  Fac. Ciencias, Avd. Los Castros s/n, 39005 Santander, Spain}
\email{beltranc@unican.es}

\begin{abstract}

We define a family of random sets of points, the Diamond ensemble, on the sphere $\S{2}$ depending on several parameters. Its most important property is that, for some of these parameters, the asymptotic expected value of the logarithmic energy of the points can be computed rigorously and shown to attain very small values, quite close to the conjectured minimal value.

\end{abstract}

\maketitle

\tableofcontents


\section{Introduction and main results}

Sets of points on the sphere $\mathbb{S}^{2}$ that are well-distributed in some sense conform an interesting object of study, see for example \cite{Brauchart2015293} for an interesting survey with different approaches to well-distributed points. One usually seeks for points with small {\em cap discrepancy} or maximal {\em separation distance} or, as we do in this paper, minimal {\em potential energy}.

\subsection{Riesz and logarithmic energy}
Given $s\in(0,\infty)$, the Riesz potential or $s$--energy of a set on points $\omega_{N} = \{ x_{1}, \ldots  ,x_{N} \}$ on the sphere $\mathbb{S}^{2}$ is
\begin{equation}\label{eq:Riesz}
\E{s}(\omega_{N}) = \displaystyle\sum_{i \neq j} \frac{1}{\| x_{i} - x_{j} \|^{s}}.
\end{equation} 
This energy has a physical interpretation for some particular values of $s$, i.e. for $s=1$ the Riesz energy is the Coulomb potential and in the special case $s=0$ the energy is defined by 
\[
\E{\log}(\omega_{N})=\left. \frac{d}{ds} \right|_{s=0}\E{s}(\omega_N) = \sum_{i \neq j} \log\| x_{i} - x_{j} \|^{-1} 
\]
and is related to the transfinite diameter and the capacity of the set by classical potential theory, see for example \cite{doohovskoy2011foundations}. 
\subsection{Smale's 7th problem}
Shub and Smale \cite{SS93} found a relation between the condition number (a quantity measuring the sensitivity to zero finding) of polynomials and the logarithmic energy of associated spherical points. Inspired by this relation, they proposed a problem that is nowadays known as Smale's 7th problem \cite{Sm2000}: find a constructive (and fast) way to produce $N$ points with quasioptimal logarithmic energy. More exactly, on input $N$, one must produce a set of $N$ points $\omega_N$ on the unit sphere such that
\[
\E{\log}(\omega_{N})-m_N\leq c\log N,
\]
where $c$ is some universal constant and $m_N$ is the minimum possible value of $\E{\log}$ among all collections of $N$ spherical points.

\subsection{The value of $m_N$} A major difficulty in Smale's 7th problem is that the value of $m_N$ is not even known up to precision $\log N$. A series of papers \cite{Wagner,RSZ94,Dubickas,Brauchart2008} gave upper and lower bounds for the value of $m_N$. The last word has been given in \cite{BS18} where this value is related to the minimum renormalized energy introduced in \cite{SS15} proving the existence of an $O(N)$ term. The current knowledge is:
\begin{equation}\label{eq:as}
m_N=W_{\log}(\S2)\,N^2-\frac12\,N\log N+C_{\log}\,N+o(N),
\end{equation}
where 
\begin{equation}\label{conjetura}
W_{\log}(\S2)=\frac{1}{(4\pi)^2}\int_{x,y\in\S2}\log\|x-y\|^{-1}\,d(x,y)=\frac{1}{2}-\log 2
\end{equation}
is the continuous energy and $C_{\log}$ is a constant. Combining  \cite{Dubickas} with \cite{BS18} it is known that
\[
-0.2232823526\ldots\leq C_{\log}\leq 2\log 2 +\frac12\log\frac23+3\log\frac{\sqrt\pi}{\Gamma(1/3)}=-0.0556053\ldots,
\]
and indeed the upper bound for $C_{\log}$ has been conjectured to be an equality using two different approaches \cite{BHS12,BS18}.
\subsection{Explicit constructions towards Smale's 7th problem}
Several point sequences that seem to have low logarithmic energy have been proposed. In \cite{dolomites} we find a number of families of points (some of them are random and some of them are not) together with numerical evidence of their properties for values as high as $N=50.000$ spherical points.  However, obtaining theoretical results for the properties of these sequences has proved a very hard task. Of course one can just run a generic optimization algorithm starting on those sequences and get seemingly optimal collections of points but theoretical results about the asymptotical properties of the output of such methods are quite out of reach.

Theoretical computations of the energy of constructively feasible families of points has only been done in a few cases. 

Points coming from the spherical ensemble (that can be seen after a stereographic projection as the eigenvalues of $A^{-1}B$ where $A$ and $B$ are random Gaussian matrices, see \cite{krishnapur2009}) have been proved in \cite{EJP3733} to have average logarithmic energy
\[
W_{\log}(\S2)\,N^2-\frac12\,N\log N+c_1\,N+o(N),
\]
where $c_1=\log 2-\gamma/2=0.404539348109\ldots$ (here, $\gamma$ is the Euler--Mascheroni constant). 

On the other hand, points obtained (after the stereographic projection) as zeros of certain random polynomials have been studied in \cite{ABS11} proving that the expected value of the logarithmic energy in this case is
\[
W_{\log}(\S2)\,N^2-\frac12\,N\log N+c_2\,N+o(N),
\]
where $c_2=-W_{\log}(\S2)=0.1931471805\ldots$ 

Both $c_1$ and $c_2$ are quite far from the known upper bound for $C_{\log}$ and thus also far from providing an answer to Smale's problem.

\subsection{Main result: the Diamond ensemble}

In this paper, we define a collection of random points, the {\em Diamond ensemble} $\diamond(N)$, depending on several parameters. For appropriate choices of the parameters, our construction produces families of points that very much resemble some already known families for which the asymptotic expansion of the logarithmic energy is unknown, such as the octahedral points or the zonal equal area nodes, see \cites{RSZ94, HOLHOS20141092,  dolomites}.
Indeed our paper can be seen as a follow up of \cite[Theorem 3.2]{RSZ94}.

A quasioptimal choice of these parameters is described in Section \ref{sec:qode}, we call the resulting set the {\em quasioptimal Diamond ensemble}, and its main interest is that we can prove the following bound.

\begin{thm}\label{ThmPpal2}

The expected value of the logarithmic energy of the quasioptimal Diamond ensemble described in Section \ref{sec:qode} is
	\begin{equation*}
	W_{\log}(\S2)\,N^2-\frac12\,N\log N+c_\diamond\,N+o(N),
	\end{equation*}
	where $c_\diamond=- 0.0492220914515784\ldots$ satisfies
	\begin{multline*}
	14340\,c_\diamond=19120\log239 - 2270\log227 - 1460\log73 - 265\log53
	- 1935\log43\\ - 930\log31 - 1710\log19 - 1938\log17 + 19825\log13
	+ 1750\log7 \\- 4250\log5 - 131307\log3 + 56586\log2 - 7170.
	\end{multline*}
\end{thm}
The value of the constant is thus approximately $0.0058$ far from the valued conjectured in \eqref{conjetura}. The Diamond ensemble is fully constructive: once a set of parameters is chosen, one just has to choose some uniform random numbers $\theta_1,\ldots,\theta_p\in[0,2\pi]$ and then the $N$ points are simply given by the direct formulas shown in Section \ref{sec:qode}. It is thus extremely easy to generate these sequences of points.

As one can guess from the expression of $c_\diamond$, obtaining the exact value for that constant requires the computation of a huge number of elementary integrals and derivatives and has been done using the computer algebra package {\tt Maxima}. 
Our proof of Theorem \ref{ThmPpal2} is thus, in some sense, a computer aided proof.
A more simple example (with more simple parameters) that can actually be done by hand is presented in Section \ref{sec:simple}.


\subsection{Structure of the paper}

In Section \ref{formula} we present a formula for computing the energy of the roots of unity of some parallels.
In Section \ref{diamond} we define the Diamond ensemble and through the formula of Section \ref{formula} we compute its associated logarithmic energy.
In Section \ref{examples} we present some concrete examples of the Diamond ensemble. 
In particular a simple model that can be made by hand, a more elaborated example and the quasioptimal Diamond ensemble in terms of minimizing logarithmic energy.
In this section we also give the asymptotic expansion of the logarithmic energy associated to every single example.
Section \ref{section_proofs} is devoted to proofs and
Appendix \ref{appendix} contains some bounds for the error of the trapezoidal rule.


\section{A general construction and a formula for its average logarithmic energy}\label{formula}
Fix $z\in(-1,1)$. The parallel of height $z$ in the sphere $\mathbb{S}^{2} \subset \mathbb{R}^{3}$ is simply the set of points $x\in\S2$ such that $\langle x,(0,0,1)\rangle=z$. A general construction of points can then be done as follows:
\begin{enumerate}
	\item Choose a positive integer $p$ and $z_1,\ldots,z_p\in(-1,1)$. Consider the $p$ parallels with heights $z_1,\ldots,z_p$.
	\item For each $j$, $1\leq j\leq p$, choose a number $r_j$ of points to be allocated on parallel $j$.
	\item Allocate $r_j$ points in parallel $j$ (which is a circumference) by projecting the $r_j$ roots of unity onto the circumference and rotating them by random phase $\theta_j\in[0,2\pi]$.
	\item To the already constructed collection of points, add the North and South pole $(0,0,1)$ and $(0,0,-1)$.
\end{enumerate}
We will denote this random set by $\Omega(p,r_{j},z_{j})$. Explicit formulas for this construction are easily produced: points in parallel of height $z_j$ are of the form
\begin{equation}\label{points}
\begin{split}
& x = \left( \sqrt{1 - z_{j}^{2}} \cos \theta, \sqrt{1 - z_{j}^{2}} \sin \theta, z_{j} \right) \\
\end{split}
\end{equation}
for some $\theta\in[0,2\pi]$ and thus the set we have described agrees with the following definition.
\begin{defn}\label{DefnOmegaGeneral}
	
	Let $\Omega(p,r_{j},z_{j})$ be the following set of points
	
	\begin{equation}\label{DefOmega}
	\Omega(p,r_{j},z_j) = 
	\begin{cases} 
	\mathcal{N} = (0,0,1) \\
	x_{j}^{i} = \left( \sqrt{1-z_{j}^{2}}\cos\left( \frac{2\pi i}{r_{j}} + \theta_{j} \right), \sqrt{1-z_{j}^{2}}\sin\left( \frac{2\pi i}{r_{j}}  + \theta_{j}\right), z_{j} \right) \\ 
	\mathcal{S} = (0,0,-1) 
	\end{cases}
	\end{equation}
	
	\noindent where $r_{j}$ is the number of roots of unity that we consider in the parallel $j$, $1 \leq j \leq p$ is the number of parallels, $1 \leq i \leq r_{j}$ and $0 \leq \theta_{j} < 2\pi$ is a random angle rotation in the parallel $j$.
	
\end{defn}

The following proposition is easy to prove.

\begin{prop}\label{PropAveEner}
Let $x$ be chosen uniformly and randomly in the parallel of height $z_i$ and let $y$ be chosen uniformly and randomly in the parallel of height $z_j$. The average of the logarithmic energy associated to $x$ and $y$ is
\begin{equation*}
\begin{split}
- \frac{\log \left( 1 - z_{i}z_{j} + |z_{i} - z_{j}| \right)}{2}.
\end{split}
\end{equation*}
\end{prop}
The following result follows directly from Proposition \ref{PropAveEner}.

\begin{cor}\label{CorAveEner2}
Let $x_j^i$ be as in Definition \ref{DefnOmegaGeneral}. Then,
\begin{equation*}
\begin{split}
 \Esp{\theta_{j},\theta_{k}}{-\sum_{l=1}^{r_{k}} \sum_{i=1}^{r_{j}} \log \left( || x_{j}^{i} - x_{k}^{l}|| \right)} 
 = 
- r_{j}r_{k}\frac{\log \left( 1 - z_{j}z_{k} + |z_{j} - z_{k }| \right)}{2},
\end{split}
\end{equation*}
where $\theta_j,\theta_k$ are uniformly distributed in $[0,2\pi]$.
\end{cor}

From Corollary \ref{CorAveEner2} we will prove the following result which gives us an expression for the expected logarithmic energy of the set $\Omega(p,r_{j},z_j)$.
\begin{prop}\label{PropEnergyOmegaGeneral}
The average logarithmic energy of points drawn from $\Omega(p,r_{j},z_j)$ is

\begin{multline*}
\Esp{\theta_{1},...,\theta_{p} \in \left[0,2\pi\right]^{p}}{\E{\log}(\Omega(p,r_{j},z_j))}
 =\\-2\log(2)
-\sum_{j=1}^{p}
r_{j}\left[
\log(4)
+ \frac{1}{2}  \log(1 - z_{j}^{2})
+ \log r_{j}\right]\\
- \sum_{j,k=1}^{p}  r_{j}r_{k}\frac{\log \left( 1 - z_{j}z_{k} + |z_{j} - z_{k }| \right)}{2}.
\end{multline*}

\end{prop}
It turns out that, for any fixed choice of $r_1,\ldots,r_p$, one can compute exactly the optimal choice of the heights $z_1,\ldots,z_p$.
\begin{prop}\label{Propminparallel}
Given $\{ r_{1},...,r_{p} \}$ such that $r_{i} \in \mathbb{N}$, there exists a unique set of heights $\{z_{1},\ldots,z_{p} \}$ such that $z_{1} > \ldots > z_{p}$ and $\Esp{\theta_{1},...,\theta_{p} \in \left[0,2\pi\right]^{p}}{\E{\log}(\Omega(p,r_{j},z_j))}
$ is minimized.
The heights are:
\begin{equation*}
z_{l}
=
\frac{\displaystyle\sum_{j=l+1}^{p}r_{j} - \displaystyle\sum_{j=1}^{l-1}r_{j}}{1 + \displaystyle\sum_{j=1}^{p}r_{j}}=1-\frac{1+r_l+2\sum_{j=1}^{l-1}r_j}{N-1},
\end{equation*}
where $N=2+\sum_{j=1}^{p}r_j$ is the total number of points. 
\end{prop}

From now on we will denote by $\Omega(p,r_{j})$ the set $\Omega(p,r_{j},z_{j})$ where the $z_{j}$ are chosen as in Proposition \ref{Propminparallel}.
With this choice of $z_j$ we have the main result of this section:
\begin{thm}\label{cor:nuevasuma}
Let $p=2M-1$ be an odd integer.
If $r_j=r_{p+1-j}$ and $z_j$ are chosen as in Proposition \ref{Propminparallel} we then have
	\begin{equation*}
	\begin{split}
	&
	\Esp{\theta_{1},...,\theta_{p} \in \left[0,2\pi\right]^{p}}{\E{\log}(\Omega(p,r_{j}))}
	\\
	& =
	-(N-1)\log(4)-\sum_{j=1}^{p}r_{j}\log r_{j}-(N-1)\sum_{j=1}^pr_j(1-z_j)\log(1-z_j)
	\\
	& =
	-(N-1)\log(4)+r_M\log r_M-2\sum_{j=1}^{M}r_{j}\log r_{j}\\&-(N-1)\sum_{j=1}^Mr_j(1-z_j)\log(1-z_j)-(N-1)\sum_{j=1}^Mr_j(1+z_j)\log(1+z_j).
	\end{split}
	\end{equation*}
\end{thm}


\section{The Diamond ensemble}\label{diamond}

We are now ready to define the construction that leads to Theorem \ref{ThmPpal2}. 
It amounts to choose some $r_1,\ldots,r_p$ such that the energy bound computed in Theorem \ref{cor:nuevasuma} is as low as possible and can be computed up to order $o(N)$. 
Our construction is based in the following heuristic argument.
\subsection{A heuristic argument}
Let us choose $z_1,\ldots,z_p$ in such a way that they define $p$ equidistant (for the spherical distance) parallels on the sphere. In other words,
\[
z_j=\cos\frac{j}{p+1}
\]
The distance between two consecutive parallels is $\pi/(p+1)$. 
We would like to choose $r_j$ in such a way that the distance between two consecutive points of the same parallel is approximately equal to some constant times $\pi/(p+1)$. Since parallel of height $z_j$ is a circumference of radius $\sin(j/(p+1))$, this goal is attained by setting for example
\begin{equation}\label{eq:ideal}
r_{j}
=
\frac{K_0\pi\sin \left( \frac{j\pi}{p+1} \right)}{\sin \left( \frac{\pi}{2(p+1)} \right)}
.
\end{equation}
Let us forget for a moment that this gives an impossible construction (since the $r_j$ will not be integer numbers). One can then plug in in Proposition \ref{PropEnergyOmegaGeneral} these values of $z_j$ and $r_j$. After a considerable amount of work the right--hand term in Proposition \ref{PropEnergyOmegaGeneral} can be proved to have the asymptotic expansion
\begin{equation}\label{eq:auxideal}
\begin{split}
W_{\log}(\S2)N^2
-\frac{1}{2} N \log(N)
+ \left(\frac{K_0\pi}{6}-\frac{1}{2}\log K_0-\frac{\log\pi}{2}\right)N
+ \mathfrak{o}(N),
\end{split}
\end{equation}
where $N=2+r_1+\cdots+r_p$ is the total number of points in the sphere. The optimal value of $K_0$ is $K_0=3/\pi$, yielding the asymptotic
\begin{equation*}
\begin{split}
&
\Esp{\theta_{1},...\theta_{p} \in \left[0,2\pi\right]^{p}}{\E{\log}(\Omega)}
\leq
- \frac{\log\left( \frac{4}{e} \right)}{2} N^2
-\frac{1}{2} N \log(N)
+ \frac{1-\log(3)}{2} N
+ \mathfrak{o}(N),
\end{split}
\end{equation*}
where $\frac{1-\log(3)}{2} \approx -0.0493$. 
\medskip

Unfortunately, this reasoning does not actually produce collections of points since as pointed out above the number of points in each parallel must be an integer number. The computation of the formula \eqref{eq:auxideal} is done with techniques similar to the ones used below but we do not include it since we actually only use it as an inspiration of our true construction below.
\subsection{An actual construction}
Inspired on the heuristic argument above, we will try to search for sets of the form $\Omega(p,r_j)$ such that the $r_{j}$ are integer numbers close to $\frac{3\sin \left( \frac{j\pi}{p+1} \right)}{\sin \left( \frac{\pi}{2(p+1)} \right)}$. We will then choose the optimal values for the $z_j$ given by Proposition \ref{Propminparallel}. Our approach is to consider different piecewise linear approximations to the formula \eqref{eq:ideal} with $K_0=3/\pi$.

\begin{defn}
Let $p,M$ be two positive integers with $p=2M-1$ odd and let $r_{j}=r(j)$ where $r:[0,2M]\to\R$ is a continuous piecewise linear function satisfying $r(x)=r(2M-x)$ and
\begin{equation*}
r(x)
=
   \begin{cases} 
      \alpha_1+\beta_1 x              & \mbox{if } 0=t_0 \leq x \leq t_1   \\
      \vdots&\vdots\\
\alpha_n+\beta_n x              & \mbox{if } t_{n-1} \leq x \leq t_n=M
   \end{cases}
\end{equation*}
Here, $[t_0,t_1,\ldots,t_n]$ is some partition of $[0,M]$ and all the $t_\ell, \alpha_\ell,\beta_\ell$ are assumed to be integer numbers.

We assume that $\alpha_1=0$, $\alpha_\ell,\beta_\ell\geq0$ and $\beta_1>0$ and there exists a constant $A\geq2$ not depending on $M$ such that $\alpha_\ell\leq AM$ and $\beta_\ell\leq A$. We also assume that $t_1\geq cM$ for some $c\geq0$. Moreover, let $z_{j}$ be as defined in Proposition \ref{Propminparallel}. 

We call the set of points defined this way the \textit{Diamond ensemble} and we denote it by $\diamond (N)$, omiting in the notation the dependence on all the parameters $n$, $t_1,\ldots,t_n$, $\alpha_1,\ldots,\alpha_n$, $\beta_1,\ldots,\beta_n$. Note that the total number of points is
\[
N=2-(\alpha_n+\beta_nM)+2\sum_{\ell=1}^n\;\sum_{j=t_{\ell-1}+1}^{t_\ell}(\alpha_\ell+\beta_\ell j).
\]
We also denote by $N_\ell$ the total number of points in up to $t_{\ell-1}$, that is
\[
N_\ell=\sum_{j=1}^{t_{\ell-1}-1}r_j.
\]
Note that if $j\in[t_{\ell-1},t_\ell]$ then
\begin{multline}\label{zjs}
z_j=1-\frac{1+r_j+2\sum_{k=1}^{j-1}r_k}{N-1}=1-\frac{1+2N_j-r_j+2\sum_{k=t_{\ell-1}}^{j}(\alpha_\ell+\beta_\ell k)}{N-1}\\
=1-\frac{1+2N_j-(\alpha_\ell+\beta_\ell j)+2\alpha_\ell(j-t_{\ell-1}+1)+\beta_\ell(j+t_{\ell-1})(j-t_{\ell-1}+1)}{N-1}
\end{multline}
We thus consider the function $z(x)$ piecewise defined by the degree $2$ polynomial
\begin{multline}\label{zjs2}
z_\ell(x)=1-\frac{1+2N_j-(\alpha_\ell+\beta_\ell x)+2\alpha_\ell(x-t_{\ell-1}+1)+\beta_\ell(x+t_{\ell-1})(x-t_{\ell-1}+1)}{N-1}
\end{multline}
\end{defn}
and note that $z_j=z(j)$.
\subsection{An exact formula for the expected logarithmic energy of the Diamond ensemble}
From Theorem \ref{cor:nuevasuma}, the expected value of the log-energy of $\diamond(N)$ is given by
	\begin{equation*}
	\begin{split}
	&
	\Esp{\theta_{1},...,\theta_{p} \in \left[0,2\pi\right]^{p}}{\E{\log}(\diamond(N))}
	\\
	& =
	-(N-1)\log(4)+r(M)\log r(M)-2\sum_{j=1}^{M}r(j)\log r(j)\\&-(N-1)\sum_{j=1}^Mr(j)(1-z(j))\log(1-z(j))-(N-1)\sum_{j=1}^Mr(j)(1+z(j))\log(1+z(j)).
	\end{split}
	\end{equation*}
We write the sums as instances of a trapezoidal composite rule. Recall that for a function  $f:[a,b]\to\R$ with $a<b$ integers, the composite trapezoidal rule is
\begin{equation}\label{eq:int_sumprevio}
T_{[a,b]}(f)=\frac{f(a)+f(b)}{2}+\sum_{j=a+1}^{b-1} f(j).
\end{equation}
We then have
\begin{cor}\label{cor:otrasuma}
The expected logarithmic energy of points drawn from the Diamond ensemble equals
		\begin{equation*}
		\begin{split}
		&
		\Esp{\theta_{1},...,\theta_{p} \in \left[0,2\pi\right]^{p}}{\E{\log}(\diamond(N))} =-(N-1)\log(4)-2\sum_{\ell=1}^nT_{[t_{\ell-1},t_\ell]}(f_\ell)
		\\&-(N-1)\sum_{\ell=1}^nT_{[t_{\ell-1},t_\ell]}(g_\ell)-(N-1)\sum_{\ell=1}^nT_{[t_{\ell-1},t_\ell]}(h_\ell),
		\end{split}
		\end{equation*}
		where for $1\leq \ell\leq n$ the functions $f_\ell,g_\ell,h_\ell$ are defined in the interval $[t_\ell-1,t_\ell]$ and satisfy 
		\begin{align*}
		f_\ell(x)=&(\alpha_\ell+\beta_\ell x)\log(\alpha_\ell+\beta_\ell x)\\
		g_\ell(x)=&(\alpha_\ell+\beta_\ell x)(1-z_\ell(x))\log(1-z_\ell(x))\\
		h_\ell(x)=&(\alpha_\ell+\beta_\ell x)\left(1+z_\ell(x)\right)\log(1+z_\ell(x))\\
		\end{align*}
\end{cor}
\subsection{An asymptotic formula for the expected logarithmic energy of the Diamond ensemble}
Since $f_\ell$ is a continuous function for $1\leq\ell\leq n$, the trapezoidal rule $T_{[t_{\ell-1},t_\ell]}(f_\ell)$ approaches the integral of $f_\ell$. Moreover,
\begin{lem}\label{lem:integralfacil}
	For $1\leq \ell\leq n$ we have
	\[
	\left|T_{[t_{\ell-1},t_\ell]}(f_\ell)-\int_{t_{{\ell-1}}}^{t_\ell}f_\ell(x)\,dx\right|\leq
	(t_\ell-t_{\ell-1})3K\log(2KM)\leq 3KM\log(2KM).\]
\end{lem}
\begin{proof}
	Let $S$ be the quantity in the lemma and note that
	\[
	S\leq\sum_{j=t_{\ell-1}+1}^{t_\ell}\int_{j-1}^j\left|f_\ell(x)-\frac{f_\ell(j-1)+f_\ell(j)}{2}\right|\,dx.
	\]
	Now, for $x\in[j-1,j]$ we have
	\[
	|f_\ell(x)-f_\ell(j-1)|\leq\int_{j-1}^j|f_\ell'(t)|\,dt\leq 1+K\log(KM+KM)\leq 2K\log(2KM).
	\]
	We thus have
	\begin{multline*}
	\left|f_\ell(x)-\frac{f_\ell(j-1)+f_\ell(j)}{2}\right|\leq\\\left|f_\ell(x)-f_\ell(j-1)\right|+\left|\frac{f_\ell(j-1)-f_\ell(j)}{2}\right|\leq 3K\log(2KM).
	\end{multline*}
	The lemma follows.
\end{proof}
Indeed, we can use the classical Euler-Maclaurin formula (see for example \cite[Th. 9.26]{Kress}) for estimating the difference between the composite trapezoidal rule and the integral in the cases of $g_\ell$ and $h_\ell$. Indeed we have
\begin{lem}\label{lem:eulermclaurin}
	The following inequality holds for $1\leq \ell\leq n$:
	\begin{align*}
	\left|T_{[t_{\ell-1},t_\ell]}(g_\ell)-\int_{t_\ell-1}^{t_\ell}g_\ell(x)\,dx-\frac{g_\ell'(t_\ell)-g_\ell'(t_{\ell-1})}{12}\right|\leq\frac{C\log M}{M}
	\end{align*}
	for some constant $C>0$.
\end{lem} 
\begin{proof}
 From Lemma \ref{lem:em} it suffices to  prove that $|g_\ell'''(x)|\,dx\leq \frac{C\log M}{M}$ for some constant $C$. Now, $g_\ell=u(x)v(x)w(x)$ where $u$ is a linear mapping, $v$ is a quadratic polynomial and $w=\log v$. The Leibniz rule for the derivative of the product gives
\[
g_\ell'''=uvw'''+6u'v'w'+3u'v''w+3uv''w'+3u'vw''+3uv'w''.
\]
If $\ell=1$ then $g_\ell'''$ has a simple expression and it is easily verified that
\[
g_1'''\leq\frac{C\log M}{M^2}
\]
for some constant $C>0$. For $\ell>1$ note now that $u(x)=\alpha_\ell +\beta_\ell x$ satisfies
\[
|u|\leq CM,\quad |u'|\leq C
\]
where $C$ is some constant. Moreover, $v(x)=1-z_\ell(x)$ satisfies
\[
c\leq |v|<1,\quad |v'|\leq \frac{C}{M},\quad |v''|\leq\frac{C}{M^2}
\]
for some positive constant $c$, not depending on $M$.
Also, $w=\log v$ satisfies
\[
|w|\leq C,\quad|w'|\leq \frac{C}{M},\quad |w''|\leq \frac{C}{M^2}\quad |w'''|\leq \frac{C}{M^3}.
\]
The lemma follows.
\end{proof}

\begin{lem}\label{lem:eulermclaurin2}
	The following inequality holds for $1\leq \ell\leq n$:
	\begin{align*}
	\left|T_{[t_{\ell-1},t_\ell]}(h_\ell)-\int_{t_\ell-1}^{t_\ell}h_\ell(x)\,dx-\frac{h_\ell'(t_\ell)-h_\ell'(t_{\ell-1})}{12}\right|\leq \frac{C}{M}
	\end{align*}
	for some constant $C>0$.
\end{lem} 
\begin{proof}
The proof is almost identical to that of Lemma \ref{lem:eulermclaurin}, so we leave it to the reader.
\end{proof}
We have proved the following.
\begin{thm}\label{maintechnical}
For the Diamond ensemble we have
		\begin{equation*}
		\begin{split}
		&
		\Esp{\theta_{1},...,\theta_{p} \in \left[0,2\pi\right]^{p}}{\E{\log}(\diamond(N))} =-(N-1)\log(4)-2\sum_{\ell=1}^n\int_{t_{{\ell-1}}}^{t_\ell}f_\ell(x)\,dx
		\\&-(N-1)\sum_{\ell=1}^n\left(\int_{t_\ell-1}^{t_\ell}g_\ell(x)\,dx+\frac{g_\ell'(t_\ell)-g_\ell'(t_{\ell-1})}{12}\right)\\&-(N-1)\sum_{\ell=1}^n\left(\int_{t_\ell-1}^{t_\ell}h_\ell(x)\,dx+\frac{h_\ell'(t_\ell)-h_\ell'(t_{\ell-1})}{12}\right)+o(M^2),
		\end{split}
		\end{equation*}
		where as before for $1\leq \ell\leq n$ the functions $f_\ell,g_\ell,h_\ell$ are as in Corollary \ref{cor:otrasuma}.
\end{thm}

\subsubsection{Zonal Equal Area Nodes}

In \cite{RSZ94} Rakhmanov et al. define a diameter bounded, equal area partition of $\mathbb{S}^{2}$ consisting on two spherical caps on the south and the north pole and rectiliniar cells located on rings of parallels. 
The resemblance between our model and this model is remarkable, and even if the 
constructions are different, the points obtained seem to be really close.
Actually, both the authors in \cite{RSZ94} and ourselves try to approximate $r_{j}$ as in equation \eqref{eq:ideal} by an integer number.
The theoretical bounds we obtain here for the logarithmic energy are slightly better than the numerical bounds obtained in \cite{dolomites} for the zonal equal area nodes.

An interesting fact is that among all the algorithmically generated point sets, the generalized spiral and zonal equal area points perform the best with respect to the logarithmic energy.
\cite[Proposition 2.3.]{dolomites} claims that the  sequence of zonal equal area configurations is equidistributed and quasi-uniform.
The same kind of result can probably be stated for the Diamond ensemble.

\section{Concrete examples of the Diamond ensemble}\label{examples}

Throughout this section we are going to explore three different examples of the Diamond ensemble.
Each of them is ilustrated with two kinds of figures: a concrete example of points following the model on $\mathbb{S}^{2}$ (figures \ref{fig_sphere_simp_01}, \ref{fig_sphere_elab_01} and \ref{fig_sphere_quasi_01}) and a comparative beetwen the $r_{j}$ that define the model and the $r_{j}$ in equation \eqref{eq:ideal} with $K_{0}=3/\pi$ (figures \ref{fig_simp_02}, \ref{fig_elab_02} and \ref{fig_quasi_02}).
In figures \ref{fig_sphere_simp_01}, \ref{fig_sphere_elab_01} and \ref{fig_sphere_quasi_01} we have used different colors for points obtained from the different linear pieces defining $r(x)$.

\subsection{A simple example}\label{sec:simple}

\begin{figure}[h]
\centering
\includegraphics[width=1\linewidth]{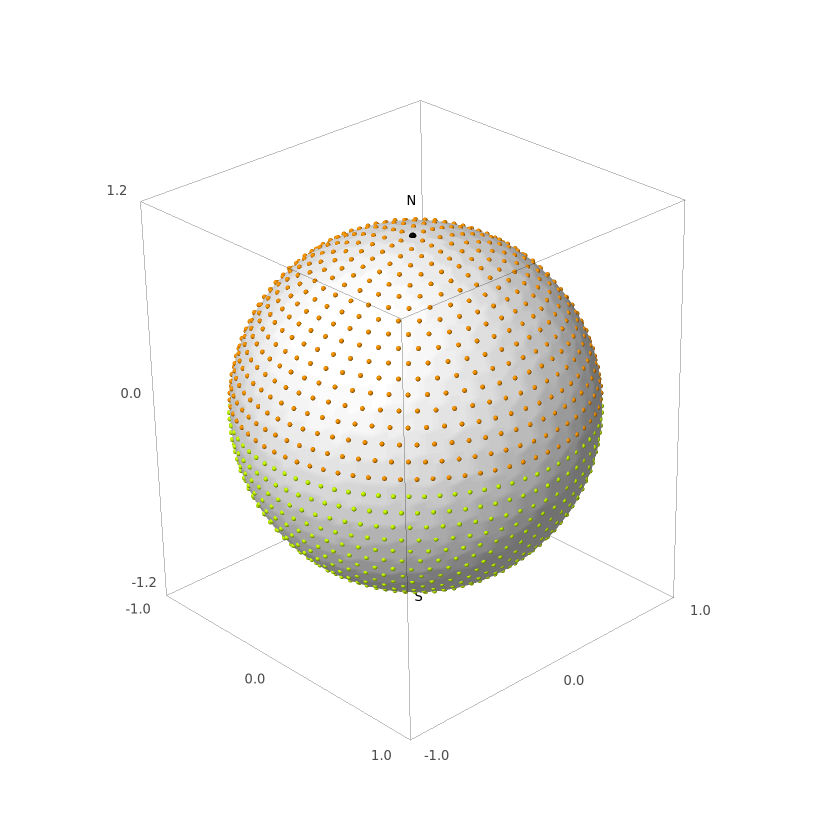}
\caption{A realization of a simple example with $K=4$ and $N=1602$.} \label{fig_sphere_simp_01}
\end{figure}	

\begin{figure}[h]
\centering
\includegraphics[width=1\linewidth]{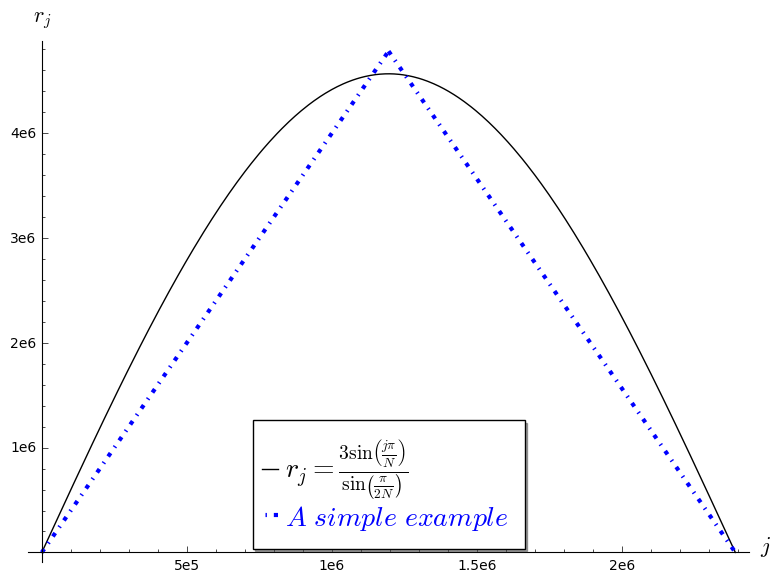}
\caption{Number of points for a simple example with $K=4$.}\label{fig_simp_02}
\end{figure}	

We choose $n=1$, $r_j=Kj$ with $K$ a positive integer for $1\leq j\leq M$. Then, for $l\in\{1,\ldots,M\}$ we have
\[
z_l=1-\frac{1+Kl^2}{N-1}.
\] 
The number of parallels is $2M-1$ and the number of points of the Diamond ensemble is 
\begin{equation*}
N
=
2 + \sum_{j=1}^{p} r_{j}
=
2 -KM+ 2\sum_{j=1}^{M} Kj
=
2 + KM^2.
\end{equation*}

One can then write down the functions in Corollary \ref{cor:otrasuma} getting
\begin{align*}
f(x)=&Kx\log(Kx).\\
g(x)=&Kx\frac{1+Kx^2}{N-1}\log\frac{1+Kx^2}{N-1}.\\
h(x)=&Kx\left(2-\frac{1+Kx^2}{N-1}\right)\log\left(2-\frac{1+Kx^2}{N-1}\right).\\
\end{align*}
Then, the formula in Theorem \ref{maintechnical} reads
\begin{equation*}
\begin{split}
&
\Esp{\theta_{1},...,\theta_{p} \in \left[0,2\pi\right]^{p}}{\E{\log}(\diamond(N))} =-(N-1)\log(4)-2\int_{0}^{M}f(x)\,dx
\\&-(N-1)\left(\int_{0}^{M}g(x)+h(x)\,dx+\frac{g'(M)-g(0)}{12}+\frac{h'(M)-h'(0)}{12}\right)+o(M^2).
\end{split}
\end{equation*}
All these integrals and derivatives can be computed, obtaining the following result.

\begin{thm}\label{thm:mainsimple}
	The expected value of the logarithmic energy of the Diamond ensemble in this section is
	\begin{multline*}
	\Esp{\theta_{1},...,\theta_{p} \in \left[0,2\pi\right]^{p}}{\E{\log}(\diamond (N))}
	=W_{\log}(\S2)N^2-\frac{1}{2}N\log N\\+N\left(\frac{\log 2}{6}K-\frac12+\log 2-\frac{\log K}{2}\right)+o(N).
	\end{multline*}
	In particular, if $K=4$ we have
	\begin{equation*}
	\Esp{\theta_{1},...,\theta_{p} \in \left[0,2\pi\right]^{p}}{\E{\log}(\diamond (N))}
	=W_{\log}(\S2)-\frac{1}{2}N\log N+N\left(\frac{2\log 2}{3}-\frac12\right)+o(N).
	\end{equation*}
\end{thm}
Note that $\frac{2\log 2}{3}-\frac12=- 0.037901879626703\ldots$ 
Using this simple example we are thus approximately $0.0177$ far from the valued conjectured in \eqref{conjetura}.


\subsubsection{Octahedral configurations of points}
In \cite{HOLHOS20141092} an area preserving map from the unit sphere to the regular octahedron is defined.	
Considering some hierarchical triangular grids on the facets of the octahedron a grid can be mapped into the sphere obtaining two different sets of points: those coming from the vertex of the grid $\Omega_{N}$ and the centers of the triangles $\Lambda_{N}$.

$\Omega_{N}$ consists on $4M^ 2 + 2$ points in the sphere that are a concrete example (with fixed angles) of our simple example.
In the paper, the authors give some numerical simulations for the logarithmic energy of this set of points that are confirmed by Theorem \ref{thm:mainsimple}.
Also in \cite[Figure 2.2]{dolomites} new numerical simulations for the same set are done obtaining a bound which is very similar to the one we prove here.
	
\subsection{A more elaborated example}

\begin{figure}[h]
\centering
\includegraphics[width=1\linewidth]{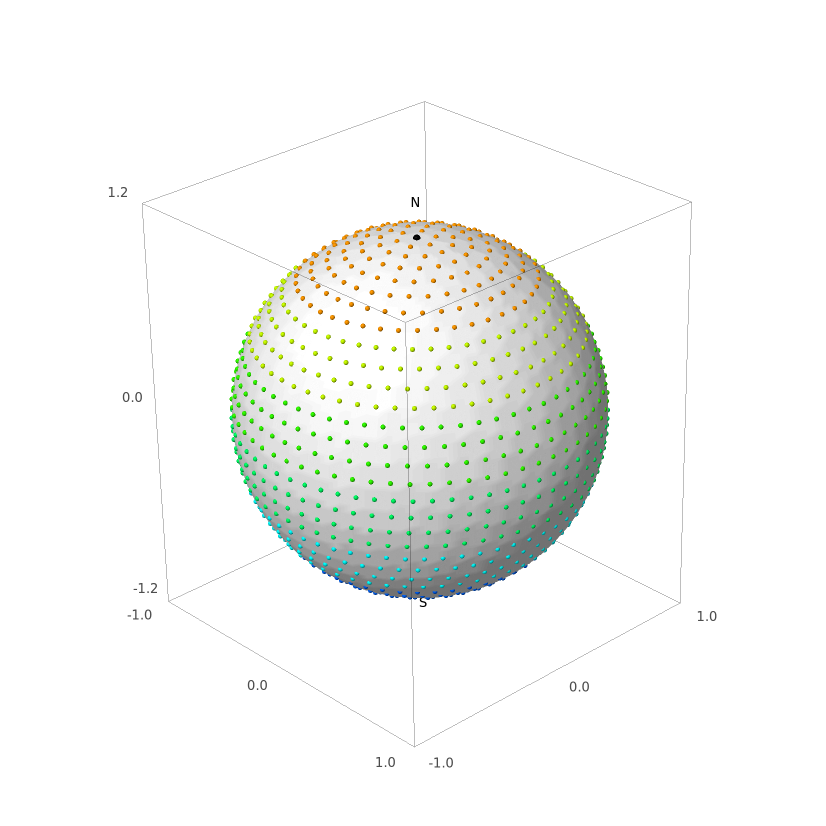}
\caption{A realization of a more elaborated example with $N=1314$.} \label{fig_sphere_elab_01}
\end{figure}	

\begin{figure}[h]
\centering
\includegraphics[width=1\linewidth]{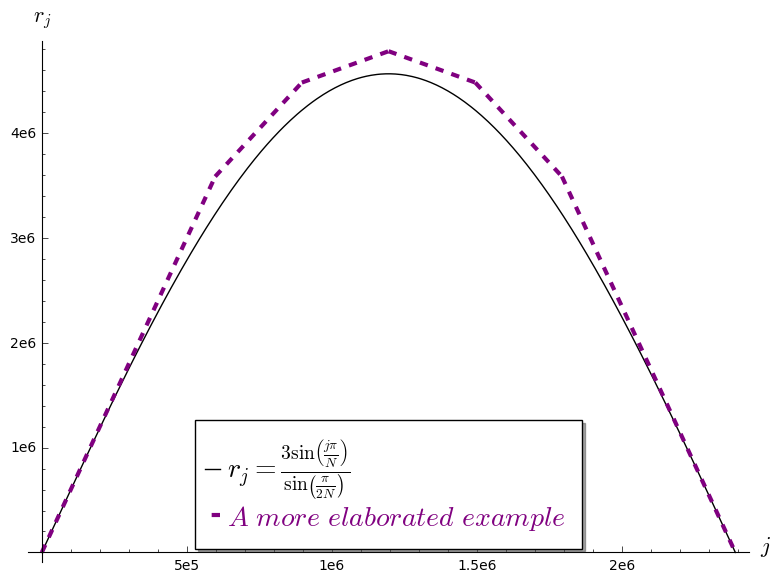}
\caption{Number of points for a more elaborated example.} \label{fig_elab_02}
\end{figure}	

The following choice of $r_j$ produces much better results. Let $p=2M-1$ where $M=4m$ with $m$ a positive integer. Let $n=3$ and let $r_j=r(j)$ where
\[
r(x)=
\begin{cases}
6x&0\leq x\leq 2m
\\6m+3x&2m\leq x\leq 3m
\\12m+x&3m\leq x\leq 4m
\\20m-x&4m\leq x\leq 5m
\\30m-3x&5m\leq x\leq 6m
\\48m-6x&6m\leq x\leq 8m
\end{cases}
\]
that satisfies $r(x)=r(p+1-x)=r(8m-x)$. Let $z_j=z(j)$ where $z(x)$ is defined by \eqref{zjs2}, that is,

\[
z(x)=
\begin{cases}
\frac{82m^2-6x^2}{82m^2+1}&0\leq x\leq 2m
\\\frac{94m^2-12mx-3x^2}{82m^2+1}&2m\leq x\leq 3m
\\\frac{112m^2-24mx-x^2}{82m^2+1}&3m\leq x\leq 4m
\\\frac{144m^2-40mx+x^2}{82m^2+1}&4m\leq x\leq 5m
\\\frac{194m^2-60mx+3x^2}{82m^2+1}&5m\leq x\leq 6m
\\\frac{302m^2-96mx+6x^2}{82m^2+1}&6m\leq x\leq 8m
\end{cases}
\]
We moreover have $N=82m^2+2$. Again, all the integrals and derivatives in Theorem \ref{maintechnical} can be computed, although this time the computer algebra package {\tt Maxima} has been used, getting the following result.

\begin{thm}\label{thm:main}

The expected value of the logarithmic energy of the Diamond ensemble in this section is
	\begin{equation*}
	\Esp{\theta_{1},...,\theta_{p} \in \left[0,2\pi\right]^{p}}{\E{\log}(\diamond (N))}
	=W_{\log}(\S{2})N^2-\frac{N}{2}\log N+cN+o(N),
	\end{equation*}
	where $c=- 0.048033870622806\ldots$ satisfies
	\begin{multline*}
	492c=- 113 \log 113 - 982 \log82 - 210 \log70 - 51 \log51\\
	+ 1638 \log41 + 900 \log15 - 36 \log12 - 1536 \log8 \\+ 144 \log6
	- 492 \log4 + 1968 \log2 - 246.
	\end{multline*}
\end{thm}
Using this more elaborated example we are thus approximately $0.0076$ far from the value conjectured in \eqref{conjetura}.
	
	
\subsection{A quasioptimal Diamond example}\label{sec:qode}

\begin{figure}[h]
\centering
\includegraphics[width=1\linewidth]{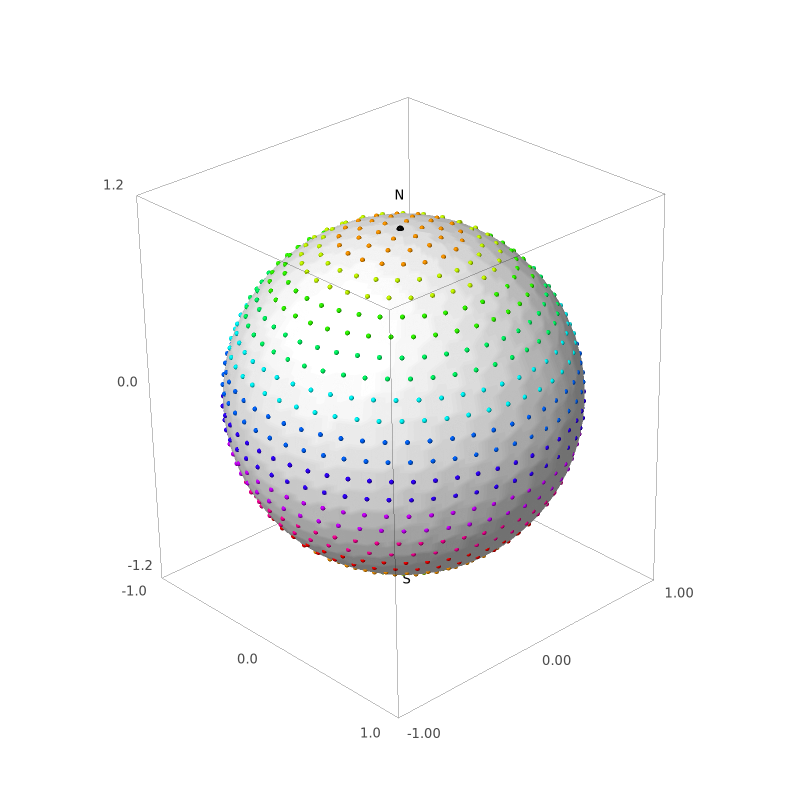}
\caption{A realization of a quasioptimal example with $N=958$.} \label{fig_sphere_quasi_01}
\end{figure}	

\begin{figure}[h]
\centering
\includegraphics[width=1\linewidth]{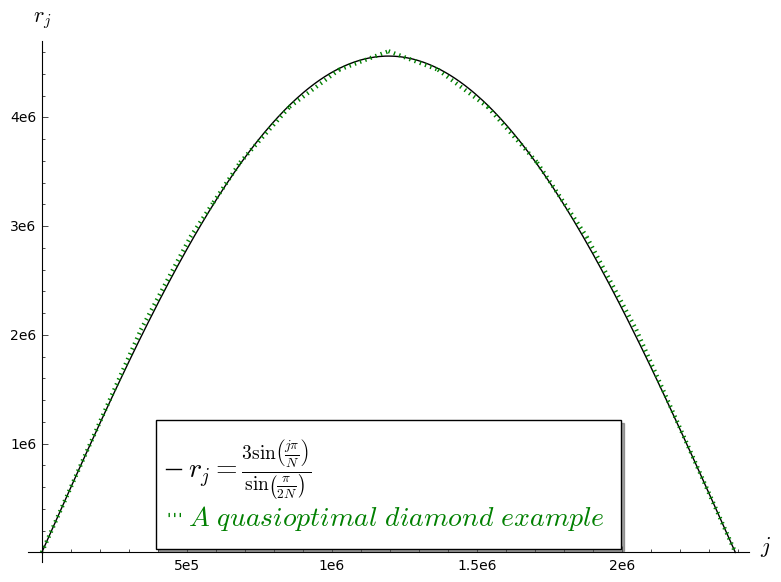}
\caption{A realization of a quasioptimal example.}\label{fig_quasi_02} 
\end{figure}	

We have made a number of tries with different choices of the parameters for the Diamond ensemble. The best one (i.e. the one with minimal logarithmic energy) that we have found is the following one: let $M=7m$ with $m$ a positive integer, let $p=2M-1$ and let
\[
r(x)=
\begin{cases}
6x&0\leq x\leq 2m
\\2m+5x&2m\leq x\leq 3m
\\5m+4x&3m\leq x\leq 4m
\\9m+3x&4m\leq x\leq 5m
\\14m+2x&5m\leq x\leq 6m
\\20m+x&6m\leq x\leq 7m
\\34m-x&7m\leq x\leq 8m
\\42m-2x&8m\leq x\leq 9m
\\51m-3x&9m\leq x\leq 10m
\\61m-4x&10m\leq x\leq 11m
\\72m-5x&11m\leq x\leq 12m
\\84m-6x&12m\leq x\leq 14m=p+1
\end{cases}
\]
that satisfies $r(x)=r(p+1-x)=r(14m-x)$. Let $z_j=z(j)$ where $z(x)$ is defined by \eqref{zjs2}, that is,
\[
z(x)=
\frac{1}{239m^2+1}\times\begin{cases}
{239m^2-6x^2}&0\leq x\leq 2m
\\{243m^2-4mx-5x^2}&2m\leq x\leq 3m
\\{252m^2-10mx-4x^2}&3m\leq x\leq 4m
\\{268m^2-18mx-3x^2}&4m\leq x\leq 5m
\\{293m^2-28mx-2x^2}&5m\leq x\leq 6m
\\{329m^2-40mx-x^2}&6m\leq x\leq 7m
\\{427m^2-68mx+x^2}&7m\leq x\leq 8m
\\{491m^2-84mx+2x^2}&8m\leq x\leq 9m
\\{572m^2-102mx+3x^2}&9m\leq x\leq 10m
\\{672m^2-122mx+4x^2}&10m\leq x\leq 11m
\\{793m^2-144mx+5x^2}&11m\leq x\leq 12m
\\{937m^2-168mx+6x^2}&12m\leq x\leq 14m=p+1
\end{cases}
\]
We moreover have $N=239m^2+2$. Again, all these integrals and derivatives have been computed by the computer algebra package {\tt Maxima}, obtaining Theorem \ref{ThmPpal2}.

\begin{figure}[h]
\centering
\includegraphics[scale=0.7]{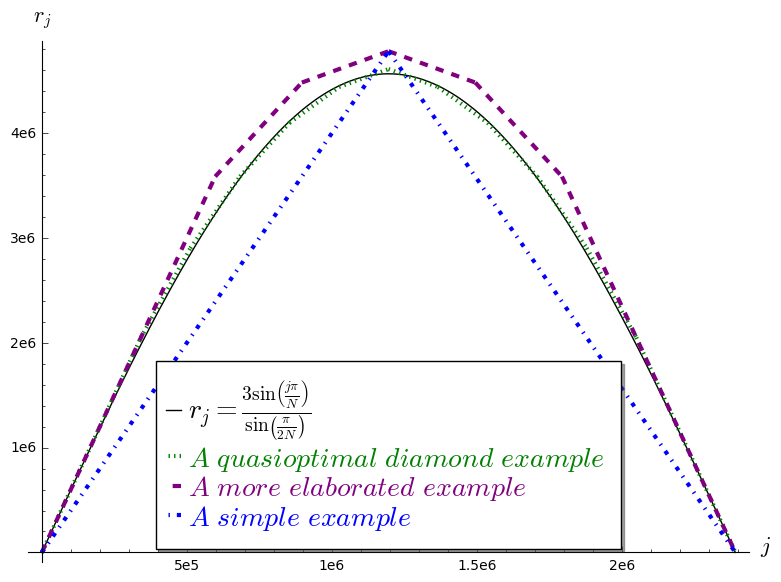}
\label{fig_r_j_simp_elab_quasi}
\caption{Comparison of the number of points in each parallel for the different models for $2M \approx 2.500.000$ parallels.}
\end{figure}


\section{Proofs of the main results}\label{section_proofs}

\subsection{Proof of Proposition \ref{PropAveEner}}

In order to prove Proposition \ref{PropAveEner}, we will need the following Lemma.

\begin{lem}\cite[Formula 4.224]{zwillinger2014table} \label{LemInt1}
	The equation
	\begin{equation*}
	\int_{0}^{\pi} \log(a + b\cos(\theta)) d\theta
	=
	\pi\log \left( \frac{a + \sqrt{a^{2} - b^{2}}}{2} \right)
	\end{equation*}
	is satisfied if $a \geq |b| > 0$.
\end{lem}

\begin{dem}
Note that
	\begin{equation*}
	\begin{split}
	& ||x-y||
	= 
	\left|\left| 
	\left( \sqrt{1 - z_{i}^{2}} \cos \theta_{i}, \sqrt{1 - z_{i}^{2}} \sin \theta_{i}, z_{i} \right) 
	-
	\left( \sqrt{1 - z_{j}^{2}} \cos \theta_{j}, \sqrt{1 - z_{j}^{2}} \sin \theta_{j}, z_{j} \right)
	\right|\right|
	\\
	& = 
	\sqrt{2}\sqrt{1 -z_{i}z_{j} - \sqrt{1 - z_{i}^{2}} \sqrt{1 - z_{j}^{2}} \cos(\theta_{i} - \theta_{j})}.
	\end{split}
	\end{equation*}
	
	\noindent We compute then
	\begin{equation*}
	\begin{split}
	& \Esp{\theta_{i},\theta_{2}}{-\log\left( \|x-y\| \right)}
	\\
	& = 
	\frac{1}{4\pi^{2}}
	\int_{0}^{2\pi} \int_{0}^{2\pi}
	- \log \left( 
	\sqrt{2}\sqrt{1 -z_{i}z_{j} - \sqrt{1 - z_{i}^{2}} \sqrt{1 - z_{j}^{2}} \cos(\theta_{i} - \theta_{j})}
	\right) d\theta_{i} d\theta_{j}
	\\
	\end{split}
	\end{equation*}
	\begin{equation*}
	\begin{split}
	& = 
	\frac{-\log(2)}{2}
	- \frac{1}{8\pi^{2}}
	\int_{0}^{2\pi} \int_{0}^{2\pi}
	\log \left(
	1 -z_{i}z_{j} - \sqrt{1 - z_{i}^{2}} \sqrt{1 - z_{j}^{2}} \cos(\theta_{i} - \theta_{j})
	\right) d\theta_{i} d\theta_{j}
	\\
	& = 
	\frac{-\log(2)}{2}
	- \frac{1}{4\pi}
	\int_{0}^{2\pi} 
	\log \left(
	1 -z_{i}z_{j} - \sqrt{1 - z_{i}^{2}} \sqrt{1 - z_{j}^{2}} \cos(\theta)
	\right) d\theta
	\\
	& = 
	\frac{-\log(2)}{2}
	- \frac{1}{2\pi}
	\int_{0}^{\pi} 
	\log \left(
	1 -z_{i}z_{j} - \sqrt{1 - z_{i}^{2}} \sqrt{1 - z_{j}^{2}} \cos(\theta)
	\right) d\theta.
	\end{split}
	\end{equation*}
	
	\noindent From Lemma \ref{LemInt1} with $a = 1 -z_{i}z_{j}$, $b = - \sqrt{1 - z_{i}^{2}} \sqrt{1 - z_{j}^{2}}$ and $z_{i} \neq z_{j}$ we have that
	\begin{equation*}
	\begin{split}
	\Esp{\theta_{i},\theta_{2}}{-\log\left( \|x-y\| \right)} 
	& = 
	\frac{-\log(2)}{2}
	- \frac{1}{2\pi}
	\pi
	\log \left( \frac{1 - z_{i}z_{j} + |z_{i} - z_{j}|}{2} \right)
	\\
	& = 
	\frac{-\log(2)}{2}
	- \frac{1}{2}
	\left[
	\log \left( 1 - z_{i}z_{j} + |z_{i} - z_{j}| \right)
	-\log(2)
	\right]
	\\
	& = 
	- \frac{\log \left( 1 - z_{i}z_{j} + |z_{i} - z_{j}| \right)}{2}.
	\end{split}
	\end{equation*}
	
	$\square$

\end{dem}

\subsection{Proof of Proposition \ref{PropEnergyOmegaGeneral}}

In order to compute de logarithmic energy associated to $\Omega(p,r_{j},z_{j})$, we have to sum the following quantities:

\begin{itemize}
	\item $A$: the energy between each point $x_{j}^{i}$, $1\leq j \leq p$ and $1\leq i \leq r_{j}$ and the north and the south pole, counted twice and the energy from the south to the north pole, again counted twice.
	\item $B$: the energy of the scaled roots of unity for every parallel $1 \leq j \leq p$.
	\item $C$: the energy between the points of every pair of parallels, as in Corollary \ref{CorAveEner2}.
\end{itemize}

\subsubsection{Computation of quantity $A$}

Note that
\begin{equation*}
\begin{split}
\|(0,0,1)-x_j^i\|=&
\sqrt{2}\sqrt{1 - z_{j}},\\
\|(0,0,-1)-x_j^i\|=&
\sqrt{2}\sqrt{1 + z_{j}}.
\end{split}
\end{equation*} 

Quantity $A$ thus equals
\begin{equation}\label{QuantityA}
\begin{split}
A=&
-2\log(2)
-\sum_{j=1}^{p} r_{j} \left( \log(4) + \log \left(  1 - z_{j}^{2}  \right) 
\right).
\end{split}
\end{equation}

\subsubsection{Computation of quantity $B$}

We will use the following results from \cite[Pg. 3]{BLMS:BLMS0621}: the logarithmic energy  associated to $N$ roots of unity in the unit circumference is $-N\log N$. As a trivial consequence, the logarithmic energy associated to $N$ points which are equidistributed in a circumference of radius $R$ is $-N\log N-N(N-1)\log R$.
	
Since the parallel at height $z_j$ is a circumference of radius $\sqrt{1-z_j^2}$, quantity $B$ equals

\begin{equation}\label{QuantityB}
\begin{split}
B& =
-\sum_{j=1}^{p}
r_{j}\log r_{j} + \frac{r_{j}(r_{j} - 1)}{2} \log (1 - z_{j}^{2}).
\end{split}
\end{equation}

\subsubsection{Computation of quantity $C$}

This has been done in Corollary \ref{CorAveEner2}:

\begin{equation}\label{QuantityC}
C= \sum_{k,j=1; k\neq j}^{p}  - r_{j}r_{k}\frac{\log \left( 1 - z_{j}z_{k} + |z_{j} - z_{k }| \right)}{2}.
\end{equation}

In order to compute the logarithmic energy associated to the set $\Omega(p,r_{j},z_{j})$ it only rest to sum the quantities \eqref{QuantityA}, \eqref{QuantityB} and \eqref{QuantityC}.

\begin{equation*}
\begin{split}
& 
\Esp{\theta_{1},...\theta_{p} \in \left[0,2\pi\right]^{p}}{\E{\log}(\Omega(p,r_{j},z_{j}))}
= 
-2\log(2)
-\sum_{j=1}^{p} r_{j} \left( \log(4) + \log \left(  1 - z_{j}^{2}  \right) 
\right) 
\\
&  -
\sum_{j=1}^{p} \left(r_{j}\log r_{j} + \frac{r_{j}(r_{j} - 1)}{2} \log (1 - z_{j}^{2})\right)
-
\sum_{j=1}^{p} \sum_{k\neq j}  r_{j}r_{k}\frac{\log \left( 1 - z_{j}z_{k} + |z_{j} - z_{k }| \right)}{2}
\\
= & 
-2\log(2)
-\sum_{j=1}^{p}
\left[
r_{j}\log(4)
+ r_{j}\log(1 - z_{j}^{2})
+ r_{j}\log r_{j}
+ \frac{r_{j}^{2}}{2} \log (1- z_{j}^{2})
\right.
\\
& \left.
- \frac{r_{j}}{2} \log (1- z_{j}^{2})
+ \sum_{k\neq j}  r_{j}r_{k}\frac{\log \left( 1 - z_{j}z_{k} + |z_{j} - z_{k }| \right)}{2}
\right]
\\
= &
-2\log(2)
-\sum_{j=1}^{p}
\left[
r_{j}\log(4)
+ \frac{r_{j}}{2}  \log(1 - z_{j}^{2})
+ r_{j}\log r_{j}
+ \sum_{k=1}^{p}  r_{j}r_{k}\frac{\log \left( 1 - z_{j}z_{k} + |z_{j} - z_{k }| \right)}{2}
\right].
\end{split}
\end{equation*}

$\square$


\subsection{Proof of Proposition \ref{Propminparallel}}

We derivate the formula from Proposition \ref{PropEnergyOmegaGeneral} for $z_{l}$ obtaining:

\begin{equation*}
\begin{split}
&
\frac{\partial \Esp{\theta_{1},...,\theta_{p} \in \left[0,2\pi\right]^{p}}{\E{\log}(\Omega(p,r_{j},z_{j}))}
}{\partial z_{l}}
\\
& =
\frac{\partial}{\partial z_{l}}
\left(
-\sum_{j=1}^{p}
\frac{r_{j}}{2}  \log(1 - z_{j}^{2})
-\sum_{j=1}^{p}\sum_{k=1}^{p}  
r_{j}r_{k}\frac{\log \left( 1 - z_{j}z_{k} + |z_{j} - z_{k }| \right)}{2}
\right) \\
& =
\frac{z_{l}r_{l}}{1-z_{l}^{2}}
+ \frac{z_{l}r_{l}^{2}}{1-z_{l}^{2}}
+ \sum_{j=1}^{l-1} r_{j}r_{l} \frac{1+z_{j}}{(1-z_{l})(1+z_{j})}
- \sum_{j=l+1}^{p} r_{j}r_{l} \frac{1-z_{j}}{(1+z_{l})(1-z_{j})}\\
& =
\frac{z_{l}r_{l}(1+r_{l})}{1-z_{l}^{2}}
+ \sum_{j=1}^{l-1}  \frac{r_{j}r_{l}}{1-z_{l}}
- \sum_{j=l+1}^{p}  \frac{r_{j}r_{l}}{1+z_{l}} \\
& =
\frac{r_{l}}{1-z_{l}^{2}}
\left(
(1+r_{l})z_{l}
+ (1+z_{l})\sum_{j=1}^{l-1}  r_{j}
- (1-z_{l})\sum_{j=l+1}^{p}  r_{j}
\right) \\
& =
\frac{r_{l}}{1-z_{l}^{2}}
\left(
z_{l}
+ \sum_{j=1}^{l-1}  r_{j}
- \sum_{j=l+1}^{p}  r_{j}
+ z_{l} \sum_{j=1}^{p}  r_{j}
\right)
.
\end{split}
\end{equation*}

We have then 
\begin{equation*}
\begin{split}
&
\frac{\partial \Esp{\theta_{1},...,\theta_{p} \in \left[0,2\pi\right]^{p}}{\E{\log}(\Omega(p,r_{j},z_{j}))}
}{\partial z_{l}}
=
0
\iff
z_{l}
\left(
1 + \sum_{j=1}^{p}  r_{j}
\right)
=
\sum_{j=l+1}^{p}  r_{j}
-
\sum_{j=1}^{l-1}  r_{j}
.
\end{split}
\end{equation*}

In other words,
\begin{equation*}
z_{l}
=
\frac{\displaystyle\sum_{j=l+1}^{p}r_{j} - \displaystyle\sum_{j=1}^{l-1}r_{j}}{1 + \displaystyle\sum_{j=1}^{p}r_{j}}.
\end{equation*}

$\square$

\subsection{Proof of Theorem \ref{cor:nuevasuma}}

To prove Theorem \ref{cor:nuevasuma} the following lemma will be useful.

\begin{lem}\label{lem3}
If $r_j=r_{p+1-j}$ and $z_j$ are chosen as in Proposition \ref{Propminparallel} we then have
	\begin{align*}
	\frac{1}{2}\sum_{j=1}^{p} \sum_{k=1}^{p}  r_{j}r_{k}\log \left( 1 - z_{j}z_{k} + |z_{j} - z_{k }| \right)=&(N-1)\sum_{j=1}^{p} r_{j}(1-z_j)\log(1 - z_{j})\\-&\sum_{j=1}^{p} r_{j}\log(1 - z_j).
	\end{align*}
\end{lem}
\begin{proof}
	Let 
	\[
	a_{j,k}=r_{j}r_{k}\log \left( 1 - z_{j}z_{k} + |z_{j} - z_{k }| \right) ,\quad 	b_{j,k}=r_{j}r_{k}\log \left( 1 + z_{j}z_{k} + |z_{j} + z_{k }| \right)
	\]
	and note that they satisfy:
	\[
	a_{j,k}=a_{k,j},\quad a_{j,p+1-k}=b_{j,k} ,\quad a_{p+1-j,k}=b_{j,k} ,\quad a_{p+1-j,p+1-k}=a_{j,k},\quad a_{M,M}=0.
	\]
	We thus have
	\[
	\sum_{j=1}^{p} \sum_{k=1}^{p}a_{j,k}=\sum_{j=1}^pa_{j,j}+\sum_{\stackrel{j,k=1}{j\neq k}}^{p} a_{j,k}=2\sum_{j=1}^{M-1}r_j^2\log(1-z_j^2)+\sum_{\stackrel{j,k=1}{j\neq k}}^{p} a_{j,k}.
	\]
	Moreover, recalling that $p=2M-1$,
	\begin{align}\label{eq:unaporaqui}
	\sum_{\stackrel{j,k=1}{j\neq k}}^{p} a_{j,k}=&\sum_{\stackrel{j,k=1}{j\neq k}}^Ma_{j,k}+\sum_{j=1}^{M} \sum_{k=M+1}^{2M-1}a_{j,k}+\sum_{j=M+1}^{2M-1} \sum_{k=1}^{M}a_{j,k}+\sum_{\stackrel{j,k=M+1}{j\neq k}}^{2M-1}a_{j,k}.
	\end{align}
	The two sums in the middle  of the right hand term in \eqref{eq:unaporaqui} can be rewritten as
	\[
	\sum_{j=1}^{M} \sum_{k=1}^{M-1}b_{j,k}+\sum_{j=1}^{M-1} \sum_{k=1}^{M}b_{j,k}=2\sum_{j=1}^{M} \sum_{k=1}^{M-1}b_{j,k},
	\]
	and using that $z_j\geq0$ for $1\leq j\leq M$ this last equals
	\begin{align*}
	2\sum_{j=1}^{M} \sum_{k=1}^{M-1}b_{j,k}=&	2\sum_{j=1}^{M} \sum_{k=1}^{M-1}r_jr_k\log(1+z_j)+	2\sum_{j=1}^{M} \sum_{k=1}^{M-1}r_jr_k\log(1+z_k)\\
	=&2\left(\sum_{k=1}^{M-1}r_k\right)\sum_{j=1}^{M} r_j\log(1+z_j)+	2\left(\sum_{j=1}^{M}r_j\right) \sum_{k=1}^{M-1}r_k\log(1+z_k)\\
	=&2\left(r_M+2\sum_{j=1}^{M-1}r_j\right) \sum_{k=1}^{M-1}r_k\log(1+z_k),
	\end{align*}
	where in the last step we have used that $z_M=0$. From \eqref{eq:unaporaqui} we then have proved that the sum in the lemma equals
	\begin{multline*}
	\sum_{j=1}^{M-1}r_j^2\log(1-z_j^2)+\left(r_M+2\sum_{j=1}^{M-1}r_j\right) \sum_{k=1}^{M-1}r_k\log(1+z_k)+\frac12\sum_{\stackrel{j,k=1}{j\neq k}}^Ma_{j,k}+\frac12\sum_{\stackrel{j,k=M+1}{j\neq k}}^{2M-1}a_{j,k}=\\	\sum_{j=1}^{M-1}r_j^2\log(1-z_j^2)+\left(r_M+2\sum_{j=1}^{M-1}r_j\right) \sum_{k=1}^{M-1}r_k\log(1+z_k)+\frac12\sum_{\stackrel{j,k=1}{j\neq k}}^Ma_{j,k}+\frac12\sum_{\stackrel{j,k=1}{j\neq k}}^{M-1}a_{j,k},
	\end{multline*}
	where we have used $a_{p+1-j,p+1-k}=a_{j,k}$. The two sums in the expression above have many common terms. We can rearrange them as follows:
	\begin{align*}
	\frac12\sum_{\stackrel{j,k=1}{j\neq k}}^Ma_{j,k}+\frac12\sum_{\stackrel{j,k=1}{j\neq k}}^{M-1}a_{j,k}=&\sum_{\stackrel{j,k=1}{j\neq k}}^{M-1}a_{j,k}+\frac12\sum_{k=1}^{M-1}a_{M,k}+\frac12\sum_{j=1}^{M-1}a_{j,M}\\
	=&2\sum_{j=1}^{M-1}\sum_{k=1}^{j-1}a_{j,k}+\sum_{k=1}^{M-1}a_{M,k}\\
	=&2\sum_{j=1}^{M-1}\sum_{k=1}^{j-1}a_{j,k}+\sum_{k=1}^{M-1}r_{M}r_{k}\log \left( 1+ z_{k } \right) ,
	\end{align*}
	where again we are using $a_{j,k}=a_{k,j}$ and $z_M=0$. All in one, we have proved that the sum in the lemma equals
	\begin{multline*}
	\sum_{j=1}^{M-1}r_j^2\log(1-z_j^2)+\left(r_M+2\sum_{j=1}^{M-1}r_j\right) \sum_{k=1}^{M-1}r_k\log(1+z_k)\\+2\sum_{j=1}^{M-1}\sum_{k=1}^{j-1}a_{j,k}+\sum_{k=1}^{M-1}r_{M}r_{k}\log \left( 1+ z_{k } \right).
	\end{multline*}
	Some little algebra then shows that
		\begin{align*}
		\frac{1}{2}\sum_{j=1}^{p} \sum_{k=1}^{p}  r_{j}r_{k}\log \left( 1 - z_{j}z_{k} + |z_{j} - z_{k }| \right)=&2\sum_{j=1}^{M-1} \sum_{k=1}^{j-1}  r_{j}r_{k}\log (1 - z_{j}),\\
		&+2\sum_{j=1}^{M-1} \sum_{k=1}^{j-1}  r_{j}r_{k}\log (1+z_k),\\
		&+\sum_{j=1}^{M-1}r_j^2\log(1-z_j^2),\\
		&+2\left(\sum_{j=1}^{M}r_j\right) \sum_{k=1}^{M-1}r_k\log(1+z_k).
		\end{align*}
		Changing the summation order and the name of the variables, the second term can be rewritten as
		\begin{multline*}
		2\sum_{j=1}^{M-2}r_j\log(1+z_j)\sum_{k=j+1}^{M-1}r_k=		
		2\sum_{j=1}^{M-2}r_j\log(1+z_j)\left(\sum_{k=1}^{M-1}r_k-r_j-\sum_{k=1}^{j-1}r_k\right)\\=2\sum_{j=1}^{M-1}r_j\log(1+z_j)\left(\frac{N-r_M}{2}-1-r_j-\sum_{k=1}^{j-1}r_k\right)
		\end{multline*}
		We have then proved:
		\begin{align*}
		\frac{1}{2}\sum_{j=1}^{p} \sum_{k=1}^{p}  r_{j}r_{k}\log \left( 1 - z_{j}z_{k} + |z_{j} - z_{k }| \right)=&2\sum_{j=1}^{M-1} r_{j}\log (1 - z_{j})\left(\sum_{k=1}^{j-1} r_{k} \right)\\
		&+2\sum_{j=1}^{M-1}r_j\log(1+z_j)\left(\frac{N-r_M}{2}-1-r_j-\sum_{k=1}^{j-1}r_k\right)\\
		&+\sum_{j=1}^{M-1}r_j^2\log(1-z_j^2)\\
		&+\left(N-2+r_M\right) \sum_{j=1}^{M-1}r_j\log(1+z_j).
		\end{align*}
		After simplification, we get
	\begin{align*}
	\frac{1}{2}\sum_{j=1}^{p} \sum_{k=1}^{p}  r_{j}r_{k}\log \left( 1 - z_{j}z_{k} + |z_{j} - z_{k }| \right)=&\sum_{j=1}^{M-1} r_{j}\left(r_j+2\sum_{k=1}^{j-1} r_{k} \right)\log(1 - z_{j})\\
	&-\sum_{j=1}^{M-1} r_{j}\left(r_j+2\sum_{k=1}^{j-1} r_{k} \right)\log (1+z_j)\\
	&+\left(2N-4\right) \sum_{j=1}^{M-1}r_j\log(1+z_j).
	\end{align*}
	Now we look at the first two terms recalling that
	\[
	z_j=1-\frac{1+r_j+2\sum_{k=1}^{j-1}r_k}{N-1}\Longrightarrow r_j+2\sum_{k=1}^{j-1}r_k=(N-1)(1-z_j)-1,
	\]
	and hence the sum in the lemma equals
	\begin{align*}
&(N-1)\sum_{j=1}^{M-1} r_{j}(1-z_j)\log(1 - z_{j})-\sum_{j=1}^{M-1} r_{j}\log(1 - z_{j})\\
&-(N-1)\sum_{j=1}^{M-1} r_{j}(1-z_j)\log (1+z_j)+\sum_{j=1}^{M-1} r_{j}\log(1 + z_{j})\\
&+\left(2N-4\right) \sum_{j=1}^{M-1}r_j\log(1+z_j),
	\end{align*}
	that is
	\begin{align*}
	&(N-1)\sum_{j=1}^{M-1} r_{j}(1-z_j)\log(1 - z_{j})-\sum_{j=1}^{M-1} r_{j}\log(1 - z_{j})\\
	&+(N-1)\sum_{j=1}^{M-1} r_{j}(1+z_j)\log (1+z_j)-\sum_{j=1}^{M-1} r_{j}\log(1 + z_{j}).\\
	\end{align*}	
	The symmetry $z_j=z_{p+1-j}$ implies that the last expression equals
		\begin{align*}
		&(N-1)\sum_{j=1}^{M-1} r_{j}(1-z_j)\log(1 - z_{j})-\sum_{j=1}^{M-1} r_{j}\log(1 - z_{j})\\
		&+(N-1)\sum_{j=M+1}^{p} r_{j}(1-z_j)\log (1-z_j)-\sum_{j=M+1}^{p} r_{j}\log(1 - z_{j}).\\
		\end{align*}
		We thus have proved (using $z_M=0$) that the sum of the lemma equals
\[
(N-1)\sum_{j=1}^{p} r_{j}(1-z_j)\log(1 - z_{j})-\sum_{j=1}^{p} r_{j}\log(1 - z_j).
\]
\end{proof}


We now finally prove Theorem \ref{cor:nuevasuma}.
From Proposition \ref{PropEnergyOmegaGeneral} and Lemma \ref{lem3} we have
	\begin{equation*}
	\begin{split}
	&
	\Esp{\theta_{1},...,\theta_{p} \in \left[0,2\pi\right]^{p}}{\E{\log}(\Omega(p,r_{j}))}
	=
	-2\log(2)
	-(N-2)\log(4)
	-\frac12\sum_{j=1}^{p} r_j  \log(1 - z_{j}^{2})
	\\
	&
	-\sum_{j=1}^{p} r_{j}\log r_{j}
	-(N-1)\sum_{j=1}^{p} r_{j}(1-z_j)\log(1 - z_{j})+\sum_{j=1}^{p} r_{j}\log(1 - z_j)
	.
	\end{split}
	\end{equation*}
	Now, note that using $z_j=-z_{p+1-j}$ we have
	\begin{multline*}
	\frac12\sum_{j=1}^{p} r_j  \log(1 - z_{j}^{2})=\frac12\sum_{j=1}^{p} r_j  \log(1 - z_{j})+\frac12\sum_{j=1}^{p} r_j  \log(1 + z_{j})\\\frac12\sum_{j=1}^{p} r_j  \log(1 - z_{j})+\frac12\sum_{j=1}^{p} r_j  \log(1 - z_{j})=\sum_{j=1}^{p} r_j  \log(1 - z_{j}).
	\end{multline*}
	The theorem follows.

$\square$


\appendix
\section{The error in the composite trapezoidal rule}
\label{appendix}

The following result is a well known fact in Fourier analysis.

\begin{lem}\label{lem:coefsbound}
	Let $f:[n,n+1]\to\R$ be a $C^1$ function with $n\in\mathbb{Z}$. Let $C>0$ be such that $|f'|\leq C$. Then, for all $k \geq 1$, 
	\[
	\left|\int_{n}^{n+1}\cos(2\pi kx) f(x)\,dx\right|\leq\frac{C}{2\pi k}.
	\]
\end{lem}
\begin{proof}
Integrate by parts.
\end{proof}
\begin{lem}\label{lem:em}
Let $f:[t_{\ell-1},t_\ell]\to\R$ be a $C^2$ function and assume that it is $C^3$ in the open interval with $|f'''|\leq C$. Then,
		\begin{align*}
		\left|T_{[t_{\ell-1},t_\ell]}(f)-\int_{t_\ell-1}^{t_\ell}f(x)\,dx-\frac{f'(t_\ell)-f'(t_{\ell-1})}{12}\right|\leq \frac{C(t_\ell-t_{\ell-1})}{24\pi}.
		\end{align*}
\end{lem}
	\begin{proof}
		Let $S$ be the quantity in the lemma. From the Euler-Macalaurin identity (see the version in \cite[Theorem 9.26]{Kress}), 
		\begin{multline*}
		S=\sum_{k=1}^\infty\frac{1}{2\pi^2 k^2}	\int_{t_\ell-1}^{t_\ell} \cos(2\pi k(x-t_{\ell-1}))f''(x)\,dx=\\
			\sum_{k=1}^\infty\frac{1}{2\pi^2 k^2}	\sum_{n=t_{\ell-1}}^{t_\ell-1}\int_{n}^{n+1}\cos(2\pi k(x-t_{\ell-1}))f''(x)\,dx.
		\end{multline*}
		From Lemma \ref{lem:coefsbound}, the integral inside is at most $C\pi/2$. Then,
		\[
		S\leq \sum_{k=1}^\infty\frac{t_\ell-t_{\ell-1}}{2\pi^2 k^2}\,\frac{C}{2\pi k}\leq\frac{C(t_\ell-t_{\ell-1})}{24\pi},
		\]
		as claimed.
	
	\end{proof}


\end{document}